\documentclass[journal]{IEEEtran}
\usepackage[T1]{fontenc}
\usepackage{graphicx}
\usepackage{setspace}
\usepackage{cite}
\usepackage{color}
\usepackage{amsmath,amssymb,amsfonts,amsthm}
\usepackage{bm}
\usepackage{multirow}
\usepackage{epstopdf}

\newtheorem{lemma} {Lemma}
\newtheorem{prop} {Proposition}
\newtheorem{cor} {Corollary}

\begin{document}

\title{Analysis of Information Delivery Dynamics in Cognitive Sensor Networks Using Epidemic Models}

\author{Pin-Yu Chen,~\IEEEmembership{Member,~IEEE}, Shin-Ming~Cheng,~\IEEEmembership{Member,~IEEE}, and Hui-Yu Hsu
\thanks{P.-Y. Chen is with the AI Foundations Group, IBM Thomas J. Watson Research Center, Yorktown Heights, New York, USA. Email: pin-yu.chen@ibm.com.}
\thanks{S.-M. Cheng and H.-Y. Hsu are with the Department of Computer Science and Information Engineering, National Taiwan University of Science and Technology, Taipei, Taiwan. Email: \{smcheng, m10215060\}@mail.ntust.edu.tw.}
\thanks{Copyright (c) 2012 IEEE. Personal use of this material is permitted. However, permission to use this material for any other purposes must be obtained from the IEEE by sending a request to pubs-permissions@ieee.org.}}

\maketitle

\begin{abstract}
To fully empower sensor networks with cognitive Internet of Things (IoT) technology, efficient medium access control protocols that enable the coexistence of cognitive sensor networks with current wireless infrastructure are as essential as the cognitive power in data fusion and processing due to shared wireless spectrum. Cognitive radio (CR) is introduced to increase spectrum efficiency and support such an endeavor, which thereby becomes a promising building block toward facilitating cognitive IoT. In this paper, primary users (PUs) refer to devices in existing wireless infrastructure, and secondary users (SUs) refer to cognitive sensors. For interference control between PUs and SUs, SUs adopt dynamic spectrum access and power adjustment to ensure sufficient operation of PUs, which inevitably leads to increasing latency and poses new challenges on the reliability of IoT communications.

To guarantee operations of primary systems while simultaneously optimizing system performance in cognitive radio ad hoc networks (CRAHNs), this paper proposes interference-aware flooding schemes exploiting global timeout and vaccine recovery schemes to control the heavy buffer occupancy induced by packet replications. The information delivery dynamics of SUs under the proposed interference-aware recovery-assisted flooding schemes is analyzed via epidemic models and stochastic geometry from a macroscopic view of the entire system. The simulation results show that our model can efficiently capture the complicated data delivery dynamics in CRAHNs in terms of end-to-end transmission reliability and buffer occupancy. This paper sheds new light on analysis of recovery-assisted flooding schemes in CRAHNs and provides performance evaluation of cognitive IoT services built upon CRAHNs.
\end{abstract}





\section{Introduction}
\label{sec_intro}
Cognitive Internet of Things (IoT) technology features advanced machine intelligence toward improved data sensing and analysis~\cite{Wu14,Aijaz15}, which is particularly appealing to applications involving sensor networks. To fully incorporate cognitive IoT technology in sensor networks, intelligent and efficient medium access control protocols  enabling the coexistence of sensor networks with current wireless infrastructure and simultaneously optimizing the performance of IoT applications are emergent challenges. Consequently, in addition to cognitive power in data sensing and computation, another crucial factor toward cognitive sensor networks is the cognition in spectrum access in order to fully deploy cognitive IoT.

Over the past decade, cognitive radio (CR) has received dramatic attention as it holds tremendous promise for increasing the utilization of scarce radio spectrum shared by primary (licensed) systems (PSs) and secondary (unlicensed or conditionally licensed) systems (SSs), e.g., cognitive sensor networks. In particular, recent works on cognitive IoT technology have identified CR as a critical building block that supports cognitive sensor networks~\cite{Ren16,Majumdar16,Zhang17,wu2009practical,wu2009empirical}, where PSs refer to existing wireless infrastructure and SSs refer to cognitive radio networks. Using CR terminology, secondary users (SUs) sense surrounding environment and adapt their operations around those of the primary users (PUs) to opportunistically exploit available resources while limiting their interference to PUs. In \textit{interweave} paradigm~\cite{Goldsmith09}, SUs seek and exploit the temporary spectrum opportunities without causing any interference to PUs. To further improve the spectrum usage, SUs in \textit{underlay} paradigm are allowed to concurrently transmit with PUs as long as sufficient operation of PUs is ensured. When IoT applications are built upon CR-enabled cognitive sensor networks, the information delivery dynamics are jointly affected by the activities of PUs and CR medium access protocols, which inevitably leads to increasing latency for communications among IoT devices. This paper aims to study the data delivery dynamics of flooding schemes in cognitive sensor networks, which specifies the effect of packet delivery control and dynamic spectrum access on buffer occupancy and end-to-end reliability in data transmission, and therefore provides a novel analysis framework for evaluating the performance of information delivery in cognitive sensor networks.

In ad hoc environment with one channel and slotted ALOHA MAC protocol, access probability control is an instinctive solution for spectrum access adaptation in underlay paradigm. In a nut shell, each SU tosses a coin independently with its access probability and transmits if it gets a head. By adjusting the access probability, the subset of SUs acting as interferers could be controlled to prevent violating PU's outage constraint, i.e., outage probability relative to a signal-to-interference-plus-noise ratio (SINR) threshold. To concatenate cognitive sensors as a cognitive radio ad hoc network (CRAHN), Quality of Service (QoS) provisioning for end-to-end packet transmissions is an essential must. However, the severe interference challenges conventional path-determined single-path routing toward QoS guarantee in CRAHN~\cite{Liang11,Youssef14}, for example, the destination might not receive the packet within a pre-specified period.

To tackle the aforementioned challenges, the well known \textit{flooding} scheme without pre-established path is considered as an appropriate choice~\cite{Khalife09,De09}, that is, data is propagated from the source node to the destination node with the assistance of other nodes as relays. While such solution decreases delivery delay and increases the reliability for end-to-end transmissions in CRAHN~\cite{Song14_QB2IC,Song15}, nodes might receive huge copies not destinated to themselves, which is known as the \textit{buffer occupancy} problem~\cite{Haas06}. To resolve this issue, two famous recovery schemes, global timeout~\cite{Altman13,Abreu14} and vaccine~\cite{CPY14}, have been introduced to remove unnecessary data at the intermediate nodes in order to mitigate the resource consumption for epidemic data delivery. The rationale behind the recovery scheme is that once the destination has successfully received one copy of the packet, the network nodes shall remove copies from their respective buffers and stop the packet propagation process as soon as possible.

In this paper, we introduce a novel hybrid recovery scheme combining the advantages of global timeout and vaccine schemes for the flooding in CRAHN to increase end-to-end SU transmission reliability while maintaining acceptable buffer occupancy. Two flooding schemes, static flooding and mobile flooding, are proposed to investigate the effects of mobile nodes on end-to-end SU transmission reliability~\cite{Huang14}. Motivated by the behavior that transportation of diseases much resembles data dissemination dynamics in flooding, we apply epidemic model to investigate the information delivery phenomenon, aiming at exploiting the advantages of time efficiency. Moreover, the complicated interference is carefully considered in our model by applying stochastic geometry modeling~\cite{Haenggi09} instead of the traditional interference graph~\cite{Song14}. Specifically, we could derive the end-to-end successful transmission probability in flooding while simultaneously meeting the outage constraints for both PUs and SUs.

The simulation results validate the correctness of our model in the sense that the complicated information delivery dynamics can be characterized by our schemes. In particular, we show that with the aid of mobility, the end-to-end transmission reliability can be drastically improved, which obeys the recent theoretical derivations and observations~\cite{Garetto10,Huang14}. Moreover, we observe from the simulation that by exploiting recovery scheme, our flooding schemes would significantly reduce buffer occupancy. As a result, this work serves as a powerful and efficient tool for analyzing information delivery dynamics in cognitive sensor networks.




The remainder of this paper is organized as follows. We introduce the background and survey the related works in Section~\ref{sec_work}. Section~\ref{sec_sys} presents the system model for the following analysis. Section~\ref{sec_inter} investigates interference at SUs by using stochastic geometry, and section~\ref{sec_dynamics} further analyzes the interference-aware static and mobile flooding schemes with hybrid recovery by using epidemic models. Section~\ref{sec_num} provides an analysis and the numerical validation of the proposed schemes. Finally, Section~\ref{sec_con} concludes this work.


\section{Background and Related Works}
\label{sec_work}

\subsection{Recovery Schemes}
\label{subsec_rec}
In~\cite{Haas06}, several novel ideas stemmed from biological phenomenon are proposed to delete unnecessary replication in order to 
mitigate the resource consumption due to packet circulating and replication. Two famous schemes are introduced as follows.
\begin{LaTeXdescription}
\item [Global timeout scheme.] Before the global timer expires, nodes replicate packets and forward to nodes who have not receive packet yet. Upon global timer expires, all nodes delete packets from their buffers. 
\item [Vaccine scheme.] With the aid of ``antipacket'' spreading, vaccine scheme is the most vigorous mechanism to efficiently save buffer occupancy. Specifically, once the destination receives the packet, it propagate a antipacket in an epidemic fashion. The nodes receiving antipacket will delete the packet from their buffer, join the antipacket spreading process, and upon antipacket reception they are redeemed  from infection by the same packet later. Consequently, the situation that copies of a packet may endlessly circulating through the network can be controlled. 
\end{LaTeXdescription}

\subsection{Epidemic Model}
Inspired from epidemiology, information dissemination in communication networks has been found to surprisingly resemble the transportation of epidemics~\cite{Daley01,Chen10,cheng2013diffusion,CPY14,chen2016decapitation}.  The interactive dynamics of the compartmental model can be characterized by ordinary differential equations (ODEs)~\cite{Daley01}. The most classical epidemic models such as susceptible-infected (SI) model, susceptible-infected-susceptible (SIS) model, and susceptible-infected-recovered (SIR) model are widely used in analyzing the dynamics of information dissemination over networks. Here we use SIR model to analyze the information spread among SUs in CRAHNs.

Analog to conventional SIR model, a non-informed (i.e., susceptible) node receives a packet and keeps it in the buffer as if the node is infected, and an infected node discards the packet when the packet is either outdated or successfully delivered to the destination. The latter case is as if the node is immune to the virus (an infected node recovers from the disease forever). However, it still remains open on connecting fundamental wireless communication characteristics such as interference and medium access with epidemic model, which is one of the major contributions of this paper.

\subsection{Related Work}
In the following, we review the related literatures for both end-to-end transmission reliability analysis and epidemic model with recovery schemes in CRAHNs.

Regarding analytical models for studying end-to-end features (such as reliability and delay) of routing schemes in ad hoc networks, we concentrate on the works based on stochastic geometry since a tractable closed-form solution is provided. Authors in~\cite{Xu11} proposed models to analyze the upper bound of end-to-end delay on a route in traditional ad hoc network, while only Vaze~\cite{Vaze11} investigates the complicated throughput-delay-reliability tradeoffs in such network. Jacquet \textit{et al.}~\cite{Jacquet09} further studies the delay and routing selection in opportunistic routing. The extended work~\cite{Andrews12} investigates the transmission capacity in multi-path routing. However, all mentioned literatures only consider the role of intra-system interference on the outage constraint of a node. When considering a secondary network consisting of both intra-system interference among SUs and PUs as well as inter-system interference between SUs and PUs, things become more complicated.

Our previous work~\cite{CPY10} applied stochastic geometry to model the behaviors inter-system and intra-system interference and analyzed the end-to-end transmission delay of single-path and multi-path routing schemes in CRAHN. The distributions and the lower bound of information dissemination latency in CRAHN is addressed in~\cite{Sun11} and~\cite{Cheng13}, respectively. \cite{Han13} further investigated the effects of dynamic PU traffic on the end-to-end delay. Song \textit{et al.}~\cite{Song14} focus on the delay of broadcasting in CRAHN under the awareness of topology. QB$^2$IC~\cite{Song14_QB2IC} and BRACER~\cite{Song15} proposed by Song and Xie investigate the interference-aware broadcast in terms of end-to-end reliability and delay. Moreover, mobility has been proved to facilitate end-to-end transmission reliability from the perspectives of information theory~\cite{Garetto10}. Huang \textit{et al.} further proposed a spectrum-aware mobility-assisted routing scheme~\cite{Huang14} considering channel availability in relay selection. 

Recently, epidemic models is widely utilized to evaluate the performance of data delivery in different networks of interest. In cases where networks are intermittently connected, such as delay tolerant networks~\cite{Zhang06}, a message is delivered in a \textit{store-and-forward} manner and the transportation resembles the spread of epidemics~\cite{CPY10_Globecom}, which is known as epidemic routing~\cite{Zhang07}. However, huge buffer space are wasted for replication redundancy in epidemic routing and thus controlling replication has been well studied. It can be achieved by leveraging historical encounter information~\cite{Balasubramanian07} or geographic information~\cite{Fang11,Cao14} to choose the target for relaying or to control the number of replying copies. Another branch proposed to reduce the unnecessary overhead in data dissemination is infection recovery. Haas \textit{et al.}~\cite{Haas06} propose five packet discarding methods to improve buffer occupancy and the methods are analyzed by using ODE models in~\cite{Zhang07,Altman13,CPY14}. By including social relationships among users, recovery scheme is adjusted adaptively to speed up the delivery process~\cite{Yang16}.  Moreover, Galluccio \textit{et al.}~\cite{Galluccio16} apply technique of network coding into epidemic routing to further reduce the number of ineffective transmissions between nodes. 


\section{System Model}
\label{sec_sys}

\subsection{Network Model}
In this paper, we cast cognitive sensor networks as a CRAHN, where the spatial distribution of SUs is assumed to follow a homogeneous Poisson Point Process (PPP) with density $\lambda_{SU}$. Let $\Phi_{SU}=\{Y_k\}$ denote the set of locations of the SUs and $P_{SU}$ denote the transmit power of an SU. The media access control protocol of the secondary network is assumed as slotted ALOHA with access probability $p$, i.e., each SU is allowed to be active with probability $p$. The set of active SUs is a PPP with density $p\lambda$ and is denoted as $\Phi^{SU} = \{Y_k \in \Phi_{SU} : B_k(p)=1\}$, where ${B_k(p)}$ are independent and identically distributed Bernoulli random variables with parameter $p$ associated with $Y_k$. Each SU has an outage constraint with maximum outage probability $\epsilon_{SU}$.

The CRAHN coexists with a primary ad hoc network, where the spatial distribution of primary transmitters (PTs) is assumed to follow a PPP with density $\lambda_{PT}$. Each PT has a dedicated primary receiver (PR) at distance $r_{PT}$ away with an arbitrary direction so that the PRs also form a PPP with density $\lambda_{PT}$. The set of locations of PTs is denoted as $\Phi_{PT} = \{X_i\}$ and the transmit power of a PT is $P_{PT}$. Each PR has an outage constraint with maximum outage probability $\epsilon_{PR}$.


\subsection{Channel Model}
We consider the effects of path loss attenuation, Rayleigh fading with unit average power $\mathcal{G}$ (i.e., exponentially distributed with unit mean), and background noise power $N$ in our channel model. The path-loss exponent of transmission is denoted by $\alpha$. The successful reception of a transmission at an user depends on if SINR observed by the user is larger than an SINR threshold (denoted by $\eta$).

\subsection{Flooding Scheme}
Two flooding schemes are considered in our work and the details are described below
\begin{LaTeXdescription}
\item [Static flooding.] Since complete information of end-to-end path is in general unavailable between the source and the destination SUs, the data are delivered in a store-and-forward fashion, that is, all SUs successfully receiving the data flood by source SU will participate in relaying the data to other SUs in a flooding fashion until the data is received by the destination SU.

\item [Mobile flooding.] The data transportation is similar to static flooding (i.e., once receiving the data, SU will relay in a flood fashion) except that SUs are mobile. Random direction (RD) mobility model~\cite{Groenevelt05} is applied for all SUs. SU will randomly choose a direction to travel in, a speed at which to travel, and a time duration for this travel. Each time the SU reaches the new position, flood operation is performed.
\end{LaTeXdescription}

\subsection{Hybrid Recovery Scheme}
By taking advantages of global timeout and vaccine schemes, we propose a hybrid recovery scheme, that is, a global timer $T$ is set and antipacket is spread to mitigate the buffer occupancy. In particular, when the destination node receives packet successfully, it starts flooding the antipacket. The nodes receiving antipacket will discard the packet (i.e., go to recovered state) and join the antipacket spreading process. No matter if the vaccine scheme is performed or not, all nodes will delete packets once the global timer $T$ expires. It is crucial to find an optimal $T$, which implicitly determines the system end time, to mitigate the heavy buffer occupancy issue, while providing statistical guarantees on data delivery reliability. Typically, $S(t)$, $I(t)$ and $R(t)$ are denoted as the susceptible, infected and recovered subpopulation at time instant $t$, respectively, and $S(t)+I(t)+R(t)=M+1$ is the total population consisting of $M$ nodes (excluding the destination node)  and one destination node.

\subsection{Performance Metric}
We focus on analyzing the tradeoff between reliability and buffer occupancy. Regarding reliability, the packet reception rate $P(T)$ is applied, which is defined as the probability that the destination will receive the packet before the global timer $T$ expires. Typically, we use a maximum unsuccessful reception probability $\epsilon_T$ to constraint $P(T)$, that is,
\begin{eqnarray}
P(T) \ge 1-\epsilon_T. 
\end{eqnarray}
Buffer occupancy $Q(T)$ is also considered an important metric, which is defined as the total number of infected SUs at every moments until $T$ \cite{Zhang07}, that is,
\begin{eqnarray}
Q(T) = \int_0^{T}I(t) dt.
\end{eqnarray}

\section{Analysis of Interference at SU}
\label{sec_inter}

\begin{figure}[t]
    \centering
    \includegraphics[width=3.4in]{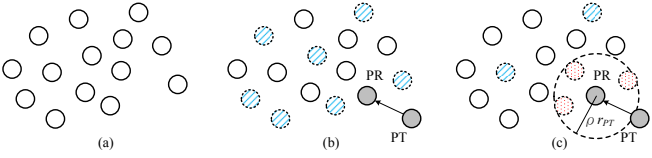}
    \caption{System model of CRAHN.~(a)~Absence of PS. The network topology is similar to traditional wireless ad hoc networks.~(b)~Outage constraint of PS. Blue slashed nodes are deactivated according to the active probability $\widetilde{p}$.~(c)~Avoidance region. Red dotted nodes are deactivated since they are in the vicinity of PR, while other blue slashed nodes are deactivated according to the active probability $\widehat{p}$.}
    \label{fig_network}
\end{figure}

\subsection{Outage Constraint of Primary Receiver Sensitivity}
To mitigate the interference to PRs, SUs exploit slotted ALOHA as the distributed spectrum access protocol. As shown in Fig. \ref{fig_network}~(b), each SU tosses a coin independently in each time slot and accesses the spectrum with head probability $\widetilde{p}$, where $\widetilde{p}$ is the parameter of i.i.d. Bernoulli random variables, $B_i(\widetilde{p})$.

\begin{lemma}
\label{lemma_density}
    To avoid interference from SUs violating the outage constraint at a PR, the permissible active density of SUs is $\widetilde{\lambda}_{SU}=\left( \frac{-\ln(1-\epsilon_{PR}) -\frac{\eta_{PR}}{P_{PT}r_{PT}^{-\alpha}}N} {r_{PT}^2 \eta_{PR}^{\delta} K_{\alpha}} - \lambda_{PT} \right)\left(\frac{P_{PT}}{P_{SU}}\right)^{\delta}$, and the active probability is $\widetilde{p}=\frac{\widetilde{\lambda}_{SU}}{\lambda_{SU}}$.
\end{lemma}

\begin{proof}
    The receiver sensitivity of a PR is maintained when only $\widetilde{p}$ portion of SUs are allowed to transmit. This subset of SUs, denoted as $\Phi_{SU}(\widetilde{p})=\{Y_i:B_i(\widetilde{p})=1\}$ with density $\widetilde{\lambda}_{SU} = \widetilde{p} \lambda_{SU}$, is obtained by independent thinning of $\Phi_{SU}$ with probability $\widetilde{p}$. We have the outage constraint
    \begin{eqnarray}
    \label{eqn3}
        \mathbb{P}\left(\frac{\mathcal{G}_{PT}P_{PT}r_{PT}^{-\alpha}}{N+ I_{SU}+I_{PT}} \geq \eta_{PR}\right) = 1 - \epsilon_{PR},
    \end{eqnarray}
    where $\mathcal{G}_{PT}$ is the channel power gain, $\eta_{PR}$ is the SINR threshold of PR, $I_{SU}=\sum_{Y_i\in \Phi_{SU}(\widetilde{p})}\mathcal{G}_{Y_i}P_{SU}\| Y_i \|^{-\alpha}$ is the interference from SUs to a typical PR, and $I_{PT}=\sum_{X_i\in \Phi_{PT}}\mathcal{G}_{X_i} P_{PT}\| X_i \|^{-\alpha}$ is the interference from other PTs. $\| \cdot \|$ denotes the distance to the origin, and the channel power gain $\mathcal{G}_{X_i}$ and $\mathcal{G}_{Y_i}$ of the interfering links are also exponentially distributed with unit mean. The left hand side of (\ref{eqn3}) can be evaluated as \cite{Ao10}
    \begin{eqnarray}
    \label{eqn01}
        &&\mathbb{P}\left[\mathcal{G}_{PT} \geq \frac{\eta_{PR}}{P_{PT}r_{PT}^{-\alpha}}(N+I_{SU}+I_{PT})\right] \nonumber\\
        &&=\exp \left(-\frac{\eta_{PR}}{P_{PT}r_{PT}^{-\alpha}}N\right) 
        \mathbb{E}\left[\exp \left(-\frac{\eta_{PR}}{P_{PT}r_{PT}^{-\alpha}}I_{SU}\right)\right] \nonumber \\
        &&~\cdot
        \mathbb{E}\left[\exp \left(-\frac{\eta_{PR}}{P_{PT}r_{PT}^{-\alpha}}I_{PT}\right)\right] \nonumber\\
        &&= \exp \left(-\frac{\eta_{PR}}{P_{PT}r_{PT}^{-\alpha}}N\right) \nonumber\\
        &&~\cdot
        \exp\left\{-\left(\widetilde{\lambda}_{SU} \left(\frac{P_{SU}}{P_{PT}}\right)^{\delta}+\lambda_{PT}\right) r_{PT}^2 \eta_{PR}^{\delta} K_{\alpha}\right\},
    \end{eqnarray}
    and from (\ref{eqn3}) and (\ref{eqn01}), when $\frac{-\ln(1-\epsilon_{PR}) -\frac{\eta_{PR}}{P_{PT}r_{PT}^{-\alpha}}N} {r_{PT}^2 \eta_{PR}^{\delta} K_{\alpha}} \geq \lambda_{PT}$ we obtain the permissible active density
    \begin{eqnarray}
    \label{eqn02}
        \widetilde{\lambda}_{SU} &=& \left( \frac{-\ln(1-\epsilon_{PR}) -\frac{\eta_{PR}}{P_{PT}r_{PT}^{-\alpha}}N} {r_{PT}^2 \eta_{PR}^{\delta} K_{\alpha}} - \lambda_{PT} \right)\left(\frac{P_{PT}}{P_{SU}}\right)^{\delta}  \nonumber \\
        &\triangleq&  \sigma P_{SU}^{-\delta},
    \end{eqnarray}
    where $K_{\alpha}=\frac{2\pi^2}{\alpha \sin(2\pi / \alpha)}$ and $\delta = 2/\alpha$.
\end{proof}

Furthermore, if a PT does not impose interference on other PRs (e.g., CDMA is exploited), additional SUs are allowed to be activated due to looser outage constraint.

\begin{cor}
\label{cor_density}
    Additional $\lambda_{PT} \left( \frac{P_{PT}}{P_{SU}} \right) ^{\delta}$ density of SUs are activated when $I_{PT}=0$.
\end{cor}
\begin{proof}
This is a direct result of (\ref{eqn02}) when $\lambda_{PT}$ is set to be $0$.
\end{proof}


\subsection{Avoidance Region}
\label{sec_avoid}
If an SU is close to a PR, deactivation of the SU (instead of following slotted ALOHA with certain access probability) may increase the overall permissible active density while maintaining the same receiver sensitivity of the PR as illustrated in Fig. \ref{fig_network}~(c). With the capability of dynamic spectrum access, an SU is deactivated when it is located within radius $\rho r_{PT}$ of a PR, where $\rho$ is a reasonably small nonnegative value named avoidance region radius coefficient. That is, each PR has an SU-avoidance region with radius $\rho r_{PT}$, and the permissible active density of SUs increases following the next lemma \cite{Ao10}.
\begin{lemma}
\label{lemma_avoid}
    The permissible active density of SUs satisfying the outage constraint of primary receiver sensitivity is enhanced from $\widetilde{\lambda}_{SU}$ to  $\widehat{\lambda}_{SU}\geq \widetilde{\lambda}_{SU}$, and $\widehat{\lambda}_{SU}= \widetilde{\lambda}_{SU} \cdot \frac{\pi/2}{\pi/2-\tan^{-1}(\sqrt{\frac{P_{PT}}{\eta_{PR}P_{SU}}}\rho^2)}$ when $\alpha=4.$
\end{lemma}

\begin{proof}
    If an SU is deactivated when it is located within radius $\rho$ of a PR, each PR has an SU-avoidance region with radius $\rho$. The probability of an SU located in $\mathcal{B}(PR;\rho)$ is $1-\exp(-\lambda_{PT}\pi\rho^2)$, where the notation $\mathcal{B}(x;r)$ represents the circle of radius $r$ centered at $x$. The interference from SUs at a typical PR located at origin becomes
\begin{eqnarray}
I'_{SU} = \sum_{Y_i \in {\Phi_{SU}^{\widehat{p}} \setminus \mathcal{B}(0;\rho)}}G_{Y_i}P_{SU}\parallel Y_i\parallel^{-\alpha} \mathbf{1}_{Y_i \not \in \mathcal{B}(PR;\rho)}.
\end{eqnarray}
Also, we have the moment generating function of $I'_{SU}$ as
\begin{eqnarray}
\label{eqn4}
\mathbb{E}[\exp(-tI'_{SU})] &=& \exp\left(-2\pi \widehat{\lambda}_{ST} \right. \nonumber \\
&~&\cdot \left. \exp(-\lambda_{PT}\pi\rho^2) \int_{\rho}^{\infty} \frac{x}{1+\frac{x^\alpha}{tP_{SU}}} dx\right).
\end{eqnarray}
Substitute $t=\frac{\eta_{PR}}{P_{PT}r_{PT}^{-\alpha}}$ and $\alpha = 4$ in (\ref{eqn4}), we have
\begin{eqnarray}
&&\mathbb{E}[\exp(-\frac{\eta_{PR}}{P_{PT}r_{PT}^{-\alpha}}I'_{SU})] \nonumber \\
&&= \exp\left\{-\widehat{\lambda}_{ST} \exp(-\lambda_{PT}\pi\rho^2) r_{PT}^2 (\frac{\eta_{PR}P_{SU}}{P_{PT}})^{\frac{1}{2}} \right. \nonumber \\
&&~\cdot \left. \left[\frac{\pi^2}{2}-\pi\tan^{-1}(\sqrt{\frac{P_{PT}}{\eta_{PR}P_{SU}}}\frac{\rho^2}{r_{PT}^2})\right] \right\}
\end{eqnarray}
Compared with the case without avoidance region, i.e., $\rho=0$ or equation (\ref{eqn02}) with $\alpha =4, K_4=\pi^2/2$, we have
\begin{eqnarray}
\frac{\widehat{\lambda}_{ST} \exp(-\lambda_{PT}\pi\rho^2)}{\widetilde{\lambda}_{SU}} = \frac{\pi^2/2}{\pi^2/2-\tan^{-1}(\sqrt{\frac{P_{PT}}{\eta_{PR}P_{SU}}}\frac{\rho^2}{r_{PT}^2})} \geq 1
\end{eqnarray}
\end{proof}


\section{Analysis of Interference-aware Information Dynamics}
\label{sec_dynamics}

\begin{figure}[t]
    \centering
    \includegraphics[width=3.4in]{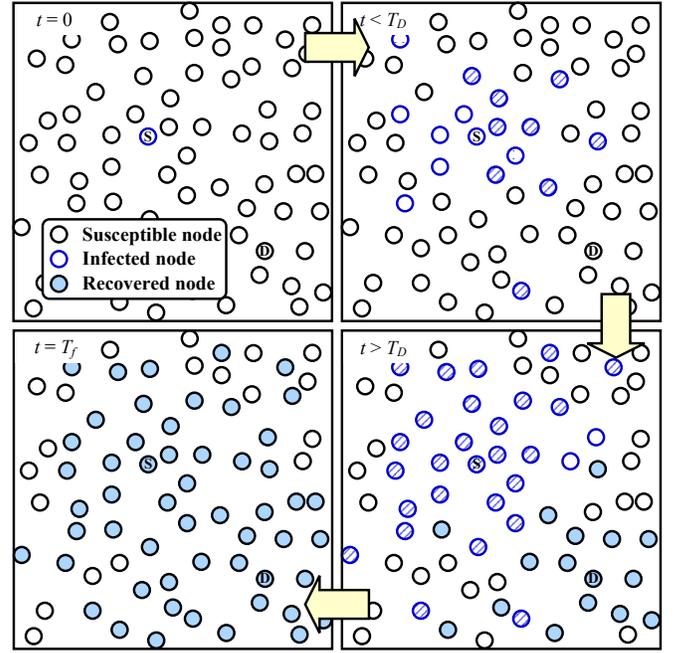}
    \caption{Illustration of the static flooding scheme. The source node (denoted as $S$) is infected at $t=0$. Susceptible nodes who successfully receive the packet from the infected nodes will transition to infected state.}
    \label{fig_static}
\end{figure}

\subsection{Static Flooding Scheme}
As shown in Fig.~\ref{fig_static}, as $t=0$, the source node at infection state starts spreading packets to all susceptible nodes. The node who successfully received the packet will be infected and join the subsequent flooding process. After the destination receives the packet (i.e., $T_D$), the hybrid recovery is performed and antipackets are spread from the destination in the flooding fashion. The node who successfully receives the antipacket (no matter currently in susceptible or infected state) will delete the packet from the buffer and transition to recovered state. In hybrid recovery scheme, only the node who is currently in infected state will join the antipacket spreading process. When the global timer expires at $T$, all infected nodes would directly transition to the recovered state and stop spreading. Taking the active density of SUs into account to satisfy the outage constraints of primary systems, we incorporate avoidance region in Section \ref{sec_avoid} to formulate the performance of static flooding scheme.

\begin{prop}
\label{prop_FER}
    For the static flooding scheme with global timeout duration $T$ and vaccine recovery scheme deployed by avoidance region in CRAHNs, the information delivery dynamics can be characterized by
    \begin{eqnarray}
    \label{eqn_con-epi}
    \left\{ \begin{array}{ll}
                \frac{dI(t)}{dt} = \widehat{p} \psi(t) \frac{\beta}{2} c \sqrt{I(t)} \frac{M-I(t)-R(t)}{M} \\ ~~~~~~~~~- \widehat{p} \mu(t) \frac{\beta}{2} c \sqrt{R(t)+P(t)} \frac{I(t)}{M}, & \\
                \frac{dR(t)}{dt} = \widehat{p} \mu(t) \frac{\beta}{2} c \sqrt{R(t)+P(t)} \frac{I(t)}{M} \\ ~~~~~~~~~+ \widehat{p} \mu(t) \frac{\beta}{2} c \sqrt{R(t)+P(t)} \frac{M-I(t)-R(t)}{M}, & \\
                \frac{dP(t)}{dt} = \widehat{p} \psi(t) \frac{\beta}{2} ~\frac{I(t)}{M}(1-P(t)), & \textnormal{~for } t \le T, \\
                I(t)=0 \textnormal{~and } R(t)=0, & \textnormal{~for } t > T, \\
            \end{array} \right.
    \end{eqnarray}
    where $\psi(t)$ is the infection rate function at time $t$, $\mu(t)$ is the recovery rate function at time $t$, $\psi(t)=\mu(t)=e^{- \beta \widehat{p}(\frac{I(t)+R(t)+P(t)}{M})T_F}$, $\beta$ is the average number of neighbors per SU, $\widehat{p}=\frac{\widehat{\lambda}_{SU}}{\lambda_{SU}}$ is the active probability of an SU, $c=2 \sqrt{\beta+1}$, and $T_F$ is the time period of a frame.
\end{prop}

\begin{proof}
    Compared with static flooding scheme in (\ref{eqn_con-epi}), the propagation of packets and antipackets is further affected by the active probability $\widehat{p}$ in CRAHNs. Assuming the traffic of an SU and its neighboring nodes follows a Poisson process with arrival rate $g=\beta \widehat{p}$ packets per unit time. The infection rate is the collision-free probability
    \begin{eqnarray}
    \psi(t)=e^{-g \cdot (\frac{I(t)+R(t)+P(t)}{M}) T_F}=e^{-\beta \widehat{p}(\frac{I(t)+R(t)+P(t)}{M})T_F}
    \end{eqnarray}
    since the interference from neighboring SUs are regarded as contentions for an SU. The same argument applies to $\mu(t)$. Suppose there is a circular region with a sufficiently large radius $R_c$, the number of SUs in the region is $M+1 = \lambda_{SU}\pi R_c^2$, and the average number of neighbors of an SU (i.e., the SUs who successfully receive the packet) can be computed as
    \begin{eqnarray}
    \label{eqn_neighbor}
        \beta &=& \lambda_{SU} \int_0^{R_c} \mathbb{P}\left(\frac{\mathcal{G}_{SU}P_{SU}r^{-\alpha}}{I_{PT}+N} \geq \eta_{SU}\right) 2\pi r dr \nonumber\\
        &\overset{R_c \rightarrow \infty}{=}& \lambda_{SU} \int_0^\infty \exp\left\{ -\frac{\eta_{SU}N}{P_{SU}}r^{\alpha}  \right. \nonumber \\
        &~&\cdot \left.
        -\lambda_{PT}\left( \frac{P_{PT}}{P_{SU}} \right)^\delta \eta_{SU}^\delta K_{\alpha} r^{2} \right\} 2\pi r dr.
    \end{eqnarray}
    From \cite[Section 2.33]{Gradshteyn00}, $\beta = k_5 {P_{SU}}^{\frac{1}{2}}$ when $\alpha=4$,
    where
    \begin{eqnarray}
    \label{eqn_k5}
        k_5 = \frac{\lambda_{SU} \pi}{2} \sqrt{\frac{\pi}{\eta_{SU}N}} \exp \left\{ \frac{{\lambda_{PT}}^2 P_{PT} \pi^4 }{16N} \right\} \mathrm{erfc}\left( \frac{{\lambda_{PT}} {P_{PT}}^{\frac{1}{2}} \pi^2 }{4N^{\frac{1}{2}}} \right) \nonumber
    \end{eqnarray}
    and $\mathrm{erfc}(x)=\frac{2}{\sqrt{\pi}} \int_x^{\infty} e^{-t^2}dt$ is the complementary error function.
\end{proof}

Note that the buffer occupancy is $Q(T)=\int^{T}_{0}I(t)dt$ by Little's formula \cite{Zhang07}. Therefore there is clearly a trade-off between global timer $T$ and buffer occupancy. With \textbf{Lemma \ref{cor_density}}, \textbf{Lemma \ref{lemma_avoid}} and (\ref{eqn_neighbor}), we then derive $\beta \widehat{p}$, the average number of active neighbors per SU.

\begin{cor}
\label{cor_neighbor}
When $\alpha=4$, the effective average number of neighbors in CRAHNs is \\
    $\beta \widehat{p}=\sigma k_5 \pi \left\{ 2 \lambda_{SU} \left[ \frac{\pi}{2}-\tan^{-1}\left( \sqrt{\frac{P_{PT}}{\eta_{PR} P_{SU}}}\rho^2\right) \right] \right\}^{-1}$.
\end{cor}

\begin{proof}
    By \textbf{Proposition 1}, $\beta = k_5 {P_{SU}}^{\frac{1}{2}}$ when $\alpha=4$, and by \textbf{Lemma 2}, $\widehat{p} = \frac{\widehat{\lambda}_{SU}}{\lambda_{SU}}$. Hence we obtain the corollary.
\end{proof}

Note that $\beta$ shall be larger than some threshold $\beta_{th}$ to prevent the network from disconnecting, which mainly depends on channel fading and path loss and is numerically analyzed in \cite{Hekmat06}, e.g., $\beta_{th}=4.52$ when there is no fading effect.
Moreover, $Q(T)$ is a function of $P_{SU}$~(or $\beta$) and $T$ since it is the integral of infected subpopulation up to time $T$. In order to minimize $Q(T)$, we formulate the optimization problem as
\begin{eqnarray}
\label{eqn_buffer_opimization}
    & \textrm{Minimize } &~Q~(P_{SU},T) \nonumber\\
    & \textrm{Subject to} &~P_{SU} \geq f^{-1}(\beta_{th}),~
    P(T) \geq 1-\epsilon_T,
\end{eqnarray}
where $\beta=f(P_{SU})$ from (\ref{eqn_neighbor}). The lower bound of $P_{SU}$ is to meet the condition that $\beta \geq \beta_{th}$ to prevent the network from disconnecting, and the packet reception rate up to time $T$ is constrained with the maximum unsuccessful reception probability $\epsilon_T$. From \textbf{Corollary \ref{cor_neighbor}}, $\beta \widehat{p}$ is a monotonically decreasing function of $P_{SU}$, decreasing $P_{SU}$ induces lower delivery delay due to increment of active SUs, and henceforth lower $T$ and $Q(T)$. When $\alpha=4$, we have $\beta=k_5 {P_{SU}}^{\frac{1}{2}} \geq \beta_{th}$ and the optimal transmission power of SU becomes $P_{SU}^* = (\beta_{th}/k_5)^2$. In addition, it is trivial that the optimal global timer $T^*$ satisfies
\begin{eqnarray}
\label{eqn_optimal_T}
P(T^*) = 1-\epsilon_T.
\end{eqnarray}
We obtain the optimal global timer $T^*$ through control the maximum unsuccessful reception probability $\epsilon_T$. The buffer occupancy can be obtained by
\begin{eqnarray}
\label{eqn_buffer_occupancy}
Q(T^*)=\int^{T^*}_{0}I(t)dt.
\end{eqnarray}
The packet reception rate and buffer occupancy can be obtained when we bring optimal global timer $T^*$ into (\ref{eqn_optimal_T}) and (\ref{eqn_buffer_occupancy}) since $Q(T^*)$ is an increasing function in $T^*$.

\begin{figure}[t]
    \centering
    \includegraphics[width=3.5in]{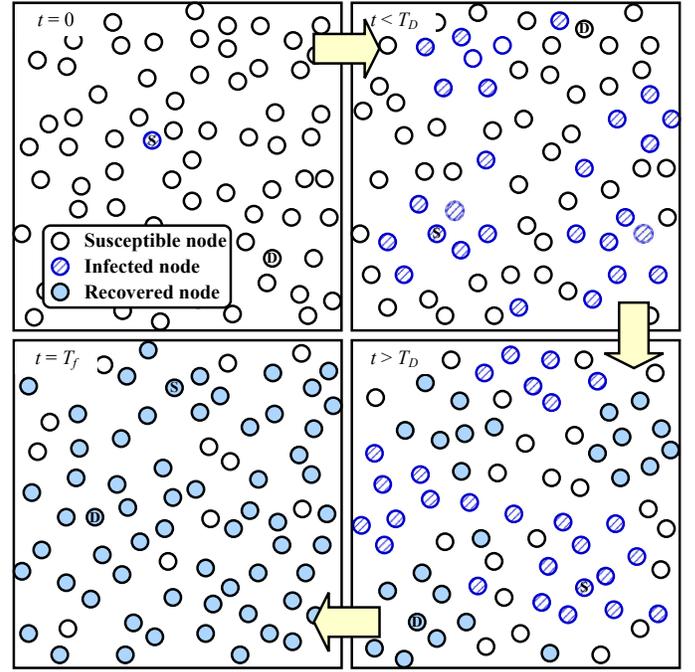}
    \caption{Illustration of the mobile flooding scheme. The source mobile node (denoted as $S$) is infected at $t=0$. Susceptible nodes who successfully receive the packet from the infected mobile nodes will transition to infected state.}
    \label{fig_mobile}
\end{figure}

\subsection{Mobile Flooding Scheme}
Figure~\ref{fig_mobile} shows the propagation behaviors of the interference-aware hybrid-recovery-assisted mobile flooding scheme. The data transportation of mobile flooding is similar to static flooding. The only difference is that nodes will move and transmit simultaneously in mobile flooding scheme.

\begin{prop}
\label{prop_MER_without_FER}
  For the mobile  flooding scheme with global timeout duration $T$ and vaccine recovery scheme deployed by avoidance region in CRAHNs, the information delivery dynamics can be characterized by
    \begin{eqnarray}
	\left\{ \begin{array}{ll}
	            \frac{dI(t)}{dt} = \widehat{p} \psi(t) \beta I(t)\frac{M-I(t)-R(t)}{M} \\ ~~~~~~~~- \widehat{p} \psi(t) \beta (R(t)+P(t)) \frac{I(t)}{M}, & \\
	            \frac{dR(t)}{dt} = \widehat{p} \mu(t) \beta (R(t)+P(t)) \frac{I(t)}{M} \\ ~~~~~~~~+ \widehat{p} \mu(t) \beta (R(t)+P(t)) \frac{M-I(t)-R(t)}{M}, & \\
	            \frac{dP(t)}{dt} = \widehat{p} \psi(t) \beta \frac{I(t)}{M}(1-P(t)), & \textnormal{~for } t \le T, \\
	            I(t)=0 \textnormal{~and } R(t)=0, & \textnormal{~for } t > T, \\
	        \end{array} \right.
	\end{eqnarray}
	  where $\psi(t)$ is the infection rate function at time $t$, $\mu(t)$ is the recovery rate function at time $t$, $\psi(t)=\mu(t)=e^{- \beta \widehat{p}(\frac{I(t)+R(t)+P(t)}{M})T_F}$, $\beta$ is the average number of neighbors per SU, and $\widehat{p}=\frac{\widehat{\lambda}_{SU}}{\lambda_{SU}}$ is the active probability of an SU.
\end{prop}
\begin{proof}
	Since mobility is shown to promote homogeneous connections or encounters among all mobile nodes in environment \cite{Daley01,Zhang07,CPY14}, the information dissemination dynamics can be captured by the homogeneous mixing SIR model in Proposition \ref{prop_MER_without_FER},  which is distinct from the propagation model for static flooding in Proposition \ref{prop_FER}. Moreover, owing to mobility, in average an active mobile node can successfully broadcast to $\beta$ nodes, and the information dissemination dynamics are further affected by the active probability $\widehat{p}$ and collision-free probability $\psi(t)$ as discussed in Proposition \ref{prop_MER_without_FER}.
\end{proof}


\section{Performance Evaluation}
\label{sec_num}

This section investigates the performance of the proposed interference-aware recovery-assisted schemes in CRAHN by using analytical models and simulation experiments. We adopt the system parameters from~\cite{Ao10,Ao12} for simulating a CRAHN. The simulator was implemented using C++ on a single machine with 4-core CPU 3.6GHz and memory of 16GB DDR3. The details of parameter setup and corresponding notations are listed in Table~\ref{table_para}. Please note that the set of system parameters has been cross-validated in~\cite{Ao10,Ao12} for efficient and reliable operation of both primary and secondary systems.


\begin{table}[!t]
\small
\caption{Simulation parameters \label{table_para}}
\centering
  \begin{tabular}{|c|c|}
  \hline
    Parameter				& 	Value    \\
    \hline
    Density of PT $\lambda_{PT}$				& 	$10^{-5}~$PUs$/$m$^{2}$    \\
    \hline
    Density of SU $\lambda_{SU}$				&	$10^{-3}~$SUs$/$m$^{2}$      \\
    \hline
    Simulation environment                      & 	$800 \times 800$ m$^2$ \\
    \hline
    Transmission power of PT $P_{PT}$ 			&	$0.3$ mW   \\
    \hline
    Transmission power of SU $P_{SU}$			&   $0.1$ mW  \\
    \hline
    Noise power $N$									&   $10^{-9}$  \\
    \hline
    SINR threshold of PR $\eta_{PR}$ and $\eta_{SU}$                      		&   $3$  \\
    \hline
    Path-loss exponent of transmission $\alpha$          							&   $4$  \\
    \hline
    Distance between PT and PR $r_{PT}$                 					&   $15$ m \\
    \hline
    Maximum outage probability of PR $\epsilon_{PR}$          					&   $0.05$	\\
    \hline
    Maximum outage probability of SU $\epsilon_{SU}$        						&   $0.1$ \\
    \hline
    Avoidance region radius coefficient $\rho$        								&   $2$ \\
    \hline
    Time period of a frame $T_F$        								&   $1$ \\
    \hline
  \end{tabular}
\end{table}

In each round of the simulation process, we distribute PUs and SUs by using PPP with means $6.4$ and $640$, respectively. Among SUs, we randomly select one source node and one destination node for the packet transmission. All SUs successfully receiving the packet flood by source SU will participate in relaying the packet to other SUs in a flooding fashion until the packet is received by the destination SU. The successful reception of a packet between two nodes depends on rules of SINR and avoidance region. Regarding the hybrid recovery, when the destination node receives packet successfully, it starts flooding the antipacket. The nodes receiving antipacket will discard the packet and join the antipacket spreading process. Moreover, all nodes will delete packets once the global timer $T$ expires. For mobile flooding, SU randomly chooses a direction with a random speed and time duration in each movement. After reaching the new position, SU perfroms the flood operation for packets and antipackets. The simulation results are collected in each round and we totally execute 20000 rounds for a particular parameter setup. The averaged results are presented as follows.

\begin{figure*}[t]
    \centering
    \includegraphics[width=6in]{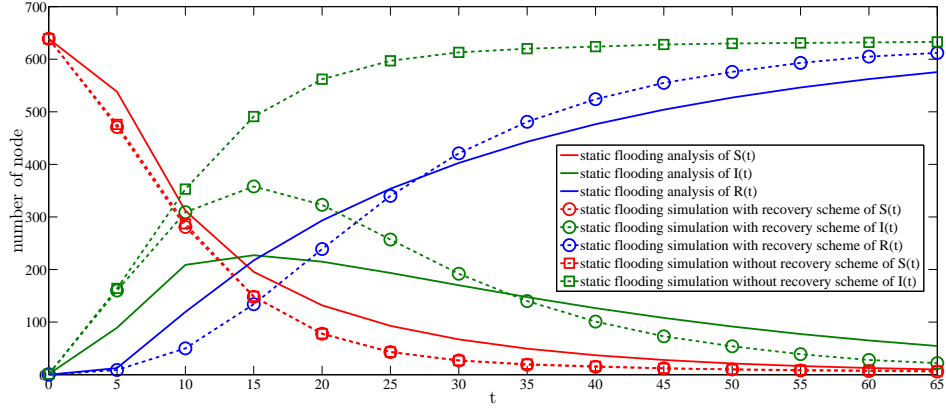}
    \caption{The number of nodes with respect to propagation time $t$ in static flooding scheme.}
    \label{fig_static_dynamics}
\end{figure*}

\begin{LaTeXdescription}
\item[Dynamics of $I(t)$, $R(t)$, and $S(t)$ in static flooding.] 
Fig. \ref{fig_static_dynamics} illustrates the number of susceptible, infected, and recovered nodes with respect to time $t$ in static flooding with the proposed hybrid recover scheme as well as the number of susceptible and infected nodes with respect to time $t$ in traditional static flooding without recovery scheme where $T=65$.  In this case, the destination node will receive the data averagely at $15$ and thus the global timer expires at the same time. We can observe that the simulation results approximately close the proposed analytical model for the hybrid recovery-assisted static flooding scheme. 

This figure also shows that the number of infected nodes increases until the destination node receives the packet. When the destination receives the packet, it performs antipacket spreading process, which enforces the nodes who receiving the antipacket transition to recovered state, thereby decreasing the number of infected nodes. Comparing with the results of traditional static flooding scheme without recovery, the number of susceptible nodes after $T_D$ decreases slightly. It is due to that nodes originally in susceptible state will transition to recovered state once they receive the antipacket. On the other hand, the number of infected nodes in recovery-assisted flooding scheme at any reference time $t>T_D$ is smaller than that in flooding without recovery scheme. It is due to two reasons. First, the nodes originally at infected state will transition to recovered state once they receive antipacket. Second, recovered nodes will not transition to infected state even they receive the packet. Consequently, the effects of antipacket spreading on the alleviation of infected nodes can be observed. 

After the global timer expires, all infected nodes transition to the recovered state, which follows (\ref{eqn_con-epi}) and assists the removing process of unnecessary packets. Obviously, the interval of global timer shall be later than the interval when the destination has received the packets. Otherwise, all nodes transition to the recovered state and the destination will receive the packet with no chance. 

\begin{figure*}[t]
    \centering
    \includegraphics[width=6in]{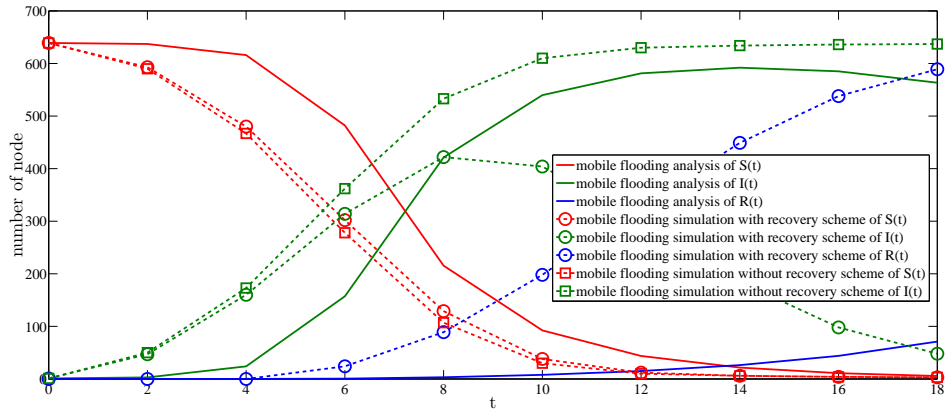}
    \caption{The number of nodes with respect to propagation time $t$ in mobile flooding scheme.}
    \label{fig_mobile_dynamics}
\end{figure*}

\item[Dynamics of $I(t)$, $R(t)$, and $S(t)$ in mobile flooding.] Fig. \ref{fig_mobile_dynamics} illustrates the number of susceptible, infected, and recovered nodes with respect to time $t$ in mobile flooding with the proposed hybrid recover scheme as well as the number of susceptible and infected nodes with respect to time $t$ in traditional mobile flooding without recovery scheme where $T=18$. In this case, the destination node will receive the data averagely at $8$ and thus the global timer expires after $T_D=8$. We can observe that the time dynamic phenomenon of susceptible, infected, and recovered nodes in the analytical model spread slowly than the simulation results and it may be the problem of parameter setup. Obviously, the spreading speed in mobile flooding is faster than that in static flooding since the advantages of mobility are fully taken.

\begin{figure}[t]
    \centering
    \includegraphics[width=3.4in]{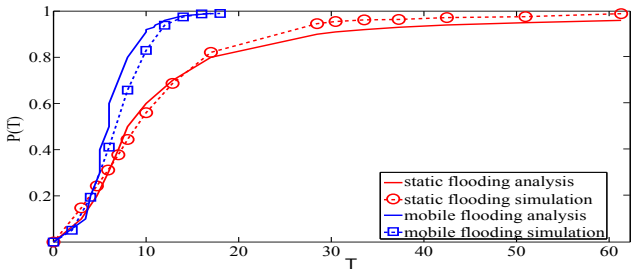}
    \caption{The packet reception probability $P(T)$ with respect to global timer $T$ in two flooding schemes.}
    \label{fig_PT}
\end{figure}


\item[Effects of $T$ on $P(T)$ in recovery-assisted static and mobile] \textbf{flooding schemes.} The global timer $T$ of two flooding schemes with respect to packet reception probability $P(T)$ is shown in Fig. \ref{fig_PT}. It is observed that $P(T)$ is a monotonic increasing function of $T$ since when $T$ is larger, the destination has higher probability to receive the packet. In addition, we find that the mobile flooding scheme has better performance than the static flooding scheme due to its mobility feature, which suggests that for a given $T$, the probability of successful packet delivery of mobile flooding scheme is higher than that of static flooding scheme. Similarly, for a fixed packet reception probability, the static flooding scheme needs to take more time than the mobile flooding scheme to achieve the same performance. Consequently, the results suggest that mobility can indeed improve data delivery.

\begin{figure}[t]
    \centering
    \includegraphics[width=3.4in]{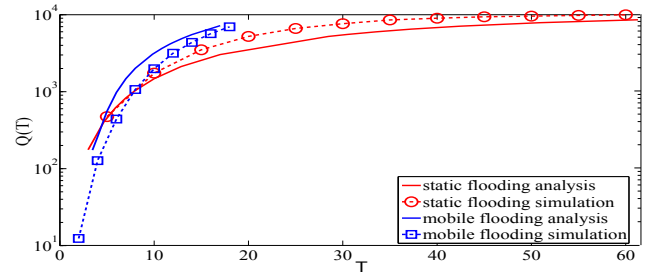}
    \caption{The Buffer occupancy $Q(T)$ with respect to global timer $T$ in two flooding schemes.}\label{fig_QT}
\end{figure}

\item[Effects of $T$ on $Q(T)$ in recovery-assisted static and mobile] \textbf{flooding schemes.} Fig. \ref{fig_QT} shows the optimal global timer $T$ of two flooding schemes with respect to buffer occupancy $Q(T)$.  We can observe that $Q(T)$ is an increasing function of $T$ since the buffer occupancy is the accumulated sum of infected nodes over time. Moreover, we also observe an intrinsic result that mobile flooding scheme spends less time to achieve same buffer occupancy compared with static flooding scheme, which can be explained by more frequent data replication. From Figs.~\ref{fig_PT} and~\ref{fig_QT}, we can observe an interesting tradeoff between reliability and buffer occupancy, that is, as $T$ increases, both $P(T)$ and $Q(T)$ increase. As a result, we can use (\ref{eqn_optimal_T}) to control the maximum packet loss rate since it well approximates the empirical packet reception probability. For instance, to achieve no more than $5 \%$ of packet loss rate, the global timer for the static and mobile flooding schemes should be $30$ and $15$, respectively. The pursuit of higher successful probability (i.e., the global timer is larger than $30$ in static flooding) will result in the price of high buffer occupancy and thus the global timer who satisfies (\ref{eqn_optimal_T}) is the optimal one.
\end{LaTeXdescription}


\section{Conclusion}
\label{sec_con}
In this paper, we study a promising architecture for cognitive sensor networks, where each sensor supporting IoT communications is equipped with CR technology for dynamic and efficient spectrum access. By casting such a cognitive sensor network as a CRAHN,
we propose a hybrid interference-aware flooding scheme for CRAHNs that utilizes global timeout and antipackets for information dissemination control. The information delivery dynamics  in CRAHNs incorporating the influences of primary receiver sensitivity, mobility of SU, and the control of recovery scheme are analyzed using a novel epidemic model. The integration of stochastic geometry and epidemic model provides efficiency and accurate analysis on reliability of end-to-end SU communications and buffer occupancy. The simulation results show that the implementation of the proposed flooding scheme indeed mitigates the buffer occupancy burden while providing statistical data delivery guarantees. Moreover, with the aid of mobility, information dissemination is shown to possess distinct characteristics that facilitates information dissemination. Consequently, this paper provides performance evaluations and modeling guidelines for efficient flooding in CRAHNs, which offers new insights on buffer occupancy and data delivery reliability analysis for cognitive IoT applications built upon CRAHNs.


\newcommand{\noop}[1]{}

\begin{IEEEbiography}[]{Pin-Yu Chen} (S'10--M'16) received the B.S. degree in electrical engineering and computer science (undergraduate honors program) from National Chiao Tung University, Taiwan, in 2009, the M.S. degree in communication engineering from National Taiwan University, Taiwan, in 2011, and the Ph.D. degree in electrical engineering and computer science and the M.A. degree in statistics from the University of Michigan Ann Arbor, MI, USA, in 2016. 
	
He is currently a research scientist in the AI Foundations Group, IBM Thomas J. Watson Research Center, Yorktown Heights, NY, USA.
His research interest includes graph and network data analytics and their applications to data mining, machine learning, artificial intelligence, and cyber security. He is a member of the \emph{Tau Beta Pi} Honor Society and the \emph{Phi Kappa Phi} Honor Society, and was the recipient of the \emph{Chia-Lun Lo Fellowship} from the University of Michigan. He received the \emph{IEEE GLOBECOM 2010 GOLD Best Paper Award} and several travel grants, including \emph{IEEE ICASSP 2014 (NSF)}, \emph{IEEE ICASSP 2015 (SPS)}, \emph{IEEE Security and Privacy Symposium 2016}, \emph{Graph Signal Processing Workshop 2016}, and \emph{ACM KDD 2016}.
\end{IEEEbiography}

\begin{IEEEbiography}[]{Shin-Ming Cheng} (S'05--M'07) received his B.S. and Ph.D. degrees in computer science and information engineering from National Taiwan University, Taipei, Taiwan, in 2000 and 2007, respectively. He was a Post-Doctoral Research Fellow at the Graduate Institute of Communication Engineering, National Taiwan University, from 2007 to 2012. Since 2012, he has been with the Department of Computer Science and Information Engineering, National Taiwan University of Science and Technology, Taipei, as an Assistant Professor. His current research interests include mobile networks, cyber security, and complex networks.

Dr. Cheng was a recipient of the \emph{IEEE PIMRC 2013 Best Paper Award} and the \emph{2014 ACM Taipei/Taiwan Chapter K. T. Li Young Researcher Award}.
\end{IEEEbiography}

\begin{IEEEbiography}[]{Hui-Yu Hsu} received her B.S. degree in computer science and information engineering from National Chi Nan University, Nantou, Taiwan, in 2013 and M.S. degree in computer science and information engineering at the National Taiwan University of Science and Technology, Taipei, Taiwan, in 2015. Her research interest includes cognitive radio networks. 
\end{IEEEbiography}

\end{document}